\documentclass{article}

\usepackage{graphicx}

\usepackage{amssymb}
\usepackage{amsthm}
\usepackage{bm}
\usepackage{currfile}
\usepackage{amsfonts}
\usepackage{amssymb,latexsym,amsmath}
\newtheorem{theorem}{Theorem}[section]
\newtheorem{example}[theorem]{Example}
\newtheorem{remark}[theorem]{Remark}
\newtheorem{corollary}[theorem]{Corollary}

\makeatletter\def\blfootnote{\xdef\@thefnmark{}\@footnotetext}\makeatother

\begin{document}




\title{Construction of zero autocorrelation stochastic waveforms\small{*}}


\author{Somantika Datta \\ Department of Mathematics, University of Idaho, \\ Moscow, Idaho, 83844-1103 USA \\ sdatta@uidaho.edu}
\date{}
\maketitle
\blfootnote{*This work was supported by AFOSR Grant No. FA9550-10-1-0441}
\begin{abstract}
Stochastic waveforms are constructed whose expected autocorrelation can be made arbitrarily small outside the origin. These waveforms are unimodular and complex-valued. Waveforms with such spike like autocorrelation are desirable in waveform design and are particularly useful in areas of radar and communications. Both discrete and continuous waveforms with low expected autocorrelation are constructed. Further, in the discrete case, frames for $\mathbb{C}^d$ are constructed from these waveforms and the frame properties of such frames are studied.
\end{abstract}









\section{Introduction}
\subsection{Motivation}
Designing unimodular waveforms with an impulse-like autocorrelation is
central in the general area of waveform design, and it is
particularly relevant in several applications in the areas of
radar and communications. In the former,  the waveforms
 can play a role in effective
target recognition, e.g., \cite{AB1}, \cite{JBJD1},
\cite{HK1}, \cite{LM1}, \cite{MLL1},  \cite{MWH1}, \cite{FEN1},
\cite{GWS1}; and in the latter they are used to address
synchronization issues in cellular (phone) access technologies,
especially code division multiple access (CDMA), e.g., \cite{UY1},
\cite{SV1}. The radar and communications methods combine in recent
advanced multifunction RF systems (AMRFS). In radar there are two main reasons that the waveforms should
be unimodular, that is, have constant amplitude. First, a
transmitter can operate at peak power if the signal has constant peak
amplitude - the system does not have to deal with the surprise of
greater than expected amplitudes. Second, amplitude variations
during transmission due to additive noise can be theoretically
eliminated. The zero autocorrelation property ensures minimum
interference between signals sharing the same channel.

Constructing unimodular waveforms with zero autocorrelation can be related to fundamental questions in harmonic analysis as follows. Let $\mathbb{R}$ be the real numbers, $\mathbb{Z}$ the integers, $\mathbb{C}$ the complex numbers, and set $\mathbb{T} = \mathbb{R}/\mathbb{Z}.$
The \textit{aperiodic autocorrelation} $A_X : \mathbb{Z} \to \mathbb{C}$
of a waveform $X : \mathbb{Z} \to \mathbb{C}$ is defined as
\begin{equation}\label{aperiodicAutocorrelation}
\forall k \in \mathbb{Z}, \quad A_X[k] = \lim_{N \to \infty}
\frac{1}{2N+1}\sum_{m = -N}^N X[k+m]\overline{X[m]}.
\end{equation}
A general problem is to characterize the
family of positive bounded Radon measures $F,$ whose inverse
Fourier transforms are the autocorrelations of bounded waveforms $X.$
A special case is when $F \equiv 1$ on
$\mathbb{T}$ and $X$ is unimodular on $\mathbb{Z}.$ This is the same as when the autocorrelation
of $X$ vanishes except at $0,$ where it takes the value $1.$ In this case $X$ is said to have \textit{perfect} autocorrelation.
An extensive discussion on the construction of different classes of \emph{deterministic} waveforms with perfect autocorrelation can be found in \cite{BD1}. Instead of aperiodic waveforms that are defined on $\mathbb{Z},$ in some applications, it might be useful to construct periodic waveforms with similar vanishing properties of the autocorrelation function.
 Let $n \geq 1$  be an integer and $\mathbb{Z}_n$ be the finite group $\{0, 1, \ldots,  n-1\}$ with addition modulo $n.$ The periodic autocorrelation $A_X:\mathbb{Z}_n \to \mathbb{C}$ of a waveform $X : \mathbb{Z}_n \to \mathbb{C}$ is defined as
 \begin{equation}\label{periodicAutocorrelation}
 \forall k= 0, 1, \ldots, n-1, \quad A_X[k] = \frac 1n \sum_{m=0}^{n-1} X[m+k]\overline{X[m]}.
  \end{equation}
 It is said that $X: \mathbb{Z}_n \to \mathbb{C}$ is a constant amplitude zero autocorrelation (CAZAC) waveform if each $|X[k]| = 1$ and
$$\forall k= 1, \ldots, n-1, \quad A_X[k] = \frac{1}{n}\sum_{m=0}^{n-1}X[m+k]\overline{X[m]} = 0.
$$
The literature on CAZACs is overwhelming. A good reference on this topic is \cite{HK1}, among many others. Literature on the general area of waveform design include \cite{Cochran05}, \cite{Bell93}, \cite{SYPMCR09}. Comparison between periodic and aperiodic autocorrelation can be found in \cite{BS2000}.

Here the focus is on the construction of stochastic aperiodic waveforms. Henceforth, the reference to waveforms shall imply aperiodic waveforms unless stated otherwise. These waveforms are stochastic in nature and are constructed from certain random variables. Due to the stochastic nature of the construction, the expected value of the corresponding autocorrelation function is analyzed. It is desired that everywhere away from zero, the expectation of the autocorrelation can be made arbitrarily small. Such waveforms will be said to have \emph{almost perfect} autocorrelation and will be called \emph{zero autocorrelation stochastic waveforms}. First discrete waveforms, $X:\mathbb{Z} \to \mathbb{C},$ are constructed such that $X$ has almost perfect autocorrelation and for all $n \in \mathbb{Z},$ $|X[n]| = 1.$ This approach is extended to the construction of continuous waveforms, $x: \mathbb{R} \to \mathbb{C}$, with similar spike like behavior of the expected autocorrelation and $|x(t)| = 1$ for all $t\in \mathbb{R}.$ Thus, these waveforms are unimodular. The stochastic and non-repetitive nature of these waveforms means that they cannot be easily intercepted or detected by an adversary. Previous work on the use of stochastic waveforms in radar can be found in \cite{Nar08}, \cite{NXHC98}, \cite{BellNar01} where the waveforms are only real-valued and not unimodular. In comparison, the waveforms constructed here are complex valued and unimodular. In addition, frame properties of frames constructed from these stochastic waveforms are discussed. This is motivated by the fact that frames have become a standard tool in signal processing. Previously, a mathematical characterization of CAZACs in terms of finite unit-normed tight frames (FUNTFs) has been done in \cite{JBJD1}.

\subsection{Notation and mathematical background}\label{Notation}
Let $X$ be a random variable with probability density function $f.$ Assuming $X$ to be absolutely continuous, the \emph{expectation} of $X,$ denoted by $E(X),$
is $$E(X) = \int_{\mathbb{R}} x f(x) \ \textrm{d}x.$$ The Gaussian random variable has probability density function given by $f(x) = \frac{1}{\sigma\sqrt{2\pi}} e^{-\frac12(\frac{x - \mu}{\sigma})^2}.$ The mean or expectation of this random variable is $\mu$ and the variance, $V(X),$ is $\sigma^2.$ In this case it is also said that $X$ follows a normal distribution and is written as $X \sim N(\mu, \sigma^2).$
The \emph{characteristic function} of $X$ at $t,$ $E(e^{itX}),$ is denoted by $\phi_{X}(t).$ For further properties of expectation and characteristic function of a random variable the reader is referred to \cite{Karr93}.

Let $\mathcal{H}$ be a Hilbert space and let $V = \{v_k, \ k \in \mathcal{K} \},$ where $\mathcal{K}$ is some index set, be a collection of vectors in $\mathcal{H}.$ Then $V$ is said to be a \emph{frame} for $\mathcal{H}$ if there exist constants $A$ and $B,$ $0< A \leq B < \infty,$ such that for any $v \in \mathcal{H}$
$$A \|v\|^2 \leq \sum_{k \in \mathcal{K}} |\langle v, v_k\rangle|^2\leq B\|v\|^2.$$
The constants $A$ and $B$ are called the \emph{frame bounds}. Thus a frame can be thought of as a redundant basis. In fact, for a finite dimensional vector space, a frame is the same as a spanning set. If $A = B,$ the frame is said to be \emph{tight}. Orthonormal bases are special cases of tight frames and for these, $A = B = 1.$

If $V$ is a frame for $\mathcal{H}$ then the map $F: \mathcal{H} \to \ell_2(\mathcal{K})$ given by $F(v) = \{\langle v, v_k\rangle \ : \ k \in \mathcal{K}\}$ is called the analysis operator. The synthesis operator is the adjoint map $F^*: \ell_2(\mathcal{K}) \to \mathcal{H},$ given by $$F^*(\{a_k\}) = \sum_{k \in \mathcal{K}} a_k v_k.$$ The \textit{frame operator} $\mathcal{F}: \mathcal{H} \to \mathcal{H}$ is given by $\mathcal{F} = F^*F.$ For a tight frame, the frame operator is just a constant multiple of the identity, i.e., $\mathcal{F} = A\mathcal{I},$ where $\mathcal{I}$ is the identity map.
Every $v \in \mathcal{H}$ can be represented as
$$v = \sum_{k \in \mathcal{K}}\langle v, \mathcal{F}^{-1} v_k \rangle v_k = \sum_{k \in \mathcal{K}}\langle v,  v_k \rangle \mathcal{F}^{-1} v_k. $$
Here $\{\mathcal{F}^{-1} v_k\}$ is also a frame and is called the dual frame.
For a tight frame, $\mathcal{F}^{-1}$ is just $\frac 1A \mathcal{F}.$ Tight frames are thus highly desirable since they offer a computationally simple reconstruction formula that does not involve inverting the frame operator.
The minimum and maximum eigenvalues of $\mathcal{F}$ are the optimal lower and upper frame bounds respectively \cite{Chr03}. Thus, for a tight frame all the eigenvalues of the frame operator are equal to each other. For the general theory on frames one can refer to \cite{Chr03}, \cite{Dau92}.

\subsection{Outline} \label{Outline}
The construction of discrete unimodular stochastic waveforms, $X: \mathbb{Z} \to \mathbb{C},$ with almost perfect autocorrelation is done in Section \ref{discretestochasticCAZAC}. This is first done with the Gaussian random variable and then generalized to other random variables. The variance of the autocorrelation is also estimated. The section also addresses the construction of stochastic waveforms in higher dimensions, i.e., construction of $v: \mathbb{Z} \to \mathbb{C}^d$ that have almost perfect autocorrelation and are unit-normed, considering the usual norm in $\mathbb{C}^d.$ In Section \ref{stochasticCAZAC} the construction of unimodular continuous waveforms with almost perfect autocorrelation is done using Brownian motion.

As mentioned in Section \ref{Notation}, frames are now a standard tool in signal processing due to their effectiveness in robust signal transmission and reconstruction. In Section \ref{StochasticFrames}, frames in $\mathbb{C}^d$ $(d \geq 2)$ are constructed from the discrete waveforms of Section \ref{discretestochasticCAZAC} and the nature of these frames is analyzed. In particular, the maximum and minimum eigenvalues of the frame operator are estimated. This helps one to understand how close these frames are to being tight. Besides, it follows, from the eigenvalue estimates, that the matrix of the analysis operator, $F,$ for such frames, can be used as a sensing matrix in compressed sensing.
\section{Construction of discrete stochastic waveforms}\label{discretestochasticCAZAC}
 In this section discrete unimodular waveforms, $X: \mathbb{Z} \to \mathbb{C},$ are constructed from random variables such that the expectation of the autocorrelation can be made arbitrarily small everywhere except at the origin. First, such a construction is done using the Gaussian random variable. Next, a general characterization of all random variables that can be used for the purpose is given.

\subsection{Construction from Gaussian random variables}\label{discrete_waveforms_Gaussian}
Let $\{Y_{\ell}\}_{\ell \in \mathbb{Z}}$ be independent identically distributed (i.i.d.) random variables following a Gaussian or normal distribution with mean $0$ and variance $\sigma^2,$ i.e., $Y_{\ell} \sim N(0, \sigma^2).$
Define $X: \mathbb{Z} \to \mathbb{C}$ by
\begin{equation} \label{eq_stoc_CAZAC}
\forall n \in \mathbb{Z}, \quad X[n] = e^{\frac{2\pi i}{\epsilon}\sum_{\ell = -n}^{n} Y_{\ell}}
\end{equation}
where $i$ is $\sqrt{-1}.$ Thus, for each $n,$ $|X[n]| = 1$ and $X$ is unimodular.
The autocorrelation of $X$ at $k \in \mathbb{Z}$ is
$$A_X[k] = \lim_{N \to \infty} \frac{1}{2N+1}\sum_{n = -N}^N X[n+k]\overline{X[n]}$$ where the limit is in the sense of probability. Theorem \ref{DiscreteGaussianWaveforms} shows that the waveform given by (\ref{eq_stoc_CAZAC}) has autocorrelation whose expectation can be made arbitrarily small for all integers $k \neq 0.$
\begin{theorem}\label{DiscreteGaussianWaveforms}
Given $\epsilon > 0,$ the waveform $X:\mathbb{Z} \to \mathbb{C}$ defined in (\ref{eq_stoc_CAZAC}) has autocorrelation $A_X$ such that
\begin{equation*}
E(A_X[k]) = \left\{ \begin{array}{ll}
1 & \textrm{if $k=0$} \\ e^{-|k|\sigma^2\left(\frac{2\pi}{\epsilon}\right)^2 } & \textrm{if $k \neq 0$}.
\end{array}
\right.
\end{equation*}
\end{theorem}
\begin{proof}
(i) When $k = 0,$ $$A_X[0] = \lim_{N \to \infty}\frac{1}{2N+1}\sum_{n = -N}^N X[n]\overline{X[n]} = 1.$$ and so $E(A_X[0]) = 1.$
\\
(ii) Let $k > 0.$ One would like to calculate
$$E(A_X[k]) = E\left(\lim_{N \to \infty} \frac{1}{2N+1} \sum_{n = -N}^N X[n+k]\overline{X[n]}\right).$$
Let $g_N(X) = \frac{1}{2N+1} \sum_{n = -N}^N X[n+k] \overline{X[n]}.$ Then $|g_N(X)| \leq 1.$ Let $h(X) = 1.$ Then for each $N,$ $|g_N(X)| \leq h(X)$ and $E[h(X)] = 1.$
 Thus, by the Dominated Convergence Theorem \cite{Karr93}, which justifies the interchange of limit and integration below, one obtains
\begin{eqnarray*}
 E(A_X[k]) &=& E\left(\lim_{N \to \infty} \frac{1}{2N+1} \sum_{n = -N}^N X[n+k]\overline{X[n]}\right) \\
&=&  \lim_{N \to \infty} \frac{1}{2N+1} \sum_{n = -N}^N E(X[n+k]\overline{X[n]}) \\
&=& \lim_{N \to \infty} \frac{1}{2N+1} \sum_{n = -N}^N E(e^{\frac{2\pi i}{\epsilon}\sum_{\ell = -n-k}^{n+k} Y_{\ell}}e^{-\frac{2\pi i}{\epsilon}\sum_{m = - n}^{n} Y_{m}}) \\
&=& \lim_{N \to \infty} \frac{1}{2N+1} \sum_{n = -N}^N E(e^{\frac{2\pi i}{\epsilon}(\sum_{\ell = n+1}^{n+k} Y_{\ell} + \sum_{m = -n-k}^{-n-1} Y_{m})}) \\
&=& \lim_{N \to \infty} \frac{1}{2N+1} \sum_{n = -N}^N \left[E\left(e^{\frac{2\pi i }{\epsilon}Y_1} \right)\right]^{2k} = \left[ E\left(e^{\frac{2\pi i }{\epsilon}Y_1} \right) \right]^{2k}= \left[\phi_{Y_1}\left(\frac{2\pi}{\epsilon}\right)\right]^{2k}
\end{eqnarray*}
where the last line uses the fact that the $Y_{\ell}$s are i.i.d. random variables. Here $\phi_{Y_1}\left(\frac{2\pi}{\epsilon}\right)$ is the characteristic function at $\frac{2\pi}{\epsilon}$ of $Y_1$ which is the same as that for any other $Y_{\ell}$ due to their identical distribution.
The characteristic function at $\frac{2\pi}{\epsilon}$ of a Gaussian random variable with mean $0$ and variance $\sigma^2$ is $e^{-\frac{\sigma^2}{2}\left(\frac{2\pi}{\epsilon}\right)^2 }.$ Thus
\begin{eqnarray*}
E(A_X[k]) &=& \left[e^{-\frac{\sigma^2}{2}\left(\frac{2\pi}{\epsilon}\right)^2 }\right]^{2k} = e^{- k \sigma^2 \left(\frac{2\pi}{\epsilon}\right)^2 }.
\end{eqnarray*}
(iii) When $k>0,$ a similar calculation for $A_X[-k]$ gives
$$E(A_X[-k]) = \left[e^{-\frac{\sigma^2}{2}\left(-\frac{2\pi}{\epsilon}\right)^2} \right]^{2k} = e^{- k \sigma^2 \left(\frac{2\pi}{\epsilon}\right)^2 }.$$
 Together, this shows that given $\epsilon$ and any $k \neq 0,$ $$E(A_X[k]) = e^{- |k|\sigma^2(\frac{2\pi}{\epsilon})^2}$$ which indicates that the expectation of the autocorrelation at any integer $k \neq 0$ can be made arbitrarily small depending on the choice of $\epsilon.$
 \end{proof}
 %

As shown in Theorem \ref{DiscreteGaussianWaveforms} the expectation of the autocorrelation can be made arbitrarily small but this is not useful unless one can estimate the variance of the autocorrelation.
 Denoting the variance of $A_X[k]$ by $V(A_X[k])$ one has
 $$ \ V(A_X[k]) = E(|A_X[k]|^2) - |E(A_X[k])|^2 = E(|A_X[k]|^2) - e^{-2|k|\sigma^2(\frac{2\pi}{\epsilon})^2}.$$
 First consider $k>0;$
 \begin{eqnarray*}
 |A_X[k]|^2 &=& \left(\lim_{N \to \infty} \frac{1}{2N+1}\sum_{n = -N}^N X[n+k]\overline{X[n]}\right) \left(\lim_{M \to \infty} \frac{1}{2M+1}\sum_{m = -M}^M \overline{X[m+k]}X[m]\right) \\
 &=& \lim_{N \to \infty} \lim_{M \to \infty} \frac{1}{(2N+1)} \frac{1}{(2M+1)} \sum_{n = -N}^N \sum_{m = -M}^M X[n+k]\overline{X[n]} \ \overline{X[m+k]}X[m] \\
 &=& \lim_{N \to \infty} \lim_{M \to \infty} \frac{1}{(2N+1)} \frac{1}{(2M+1)} \sum_{n = -N}^N \sum_{m = -M}^M \\
 && \phantom{+++++++++++++} e^{\frac{2\pi}{\epsilon} i \left(\sum_{j = 1}^k Y_{-n-j}+ \sum_{j = 1}^k Y_{n+j} - \sum_{j = 1}^k Y_{-m-j}  - \sum_{j = 1}^k Y_{m+j} \right)}\ .
 \end{eqnarray*}
 By applying the Lebesgue Dominated Convergence Theorem one can bring the expectation inside the double sum to get
 \begin{eqnarray*}
 E(|A_X[k]|^2) &=& \lim_{N \to \infty} \lim_{M \to \infty} \frac{1}{(2N+1)} \frac{1}{(2M+1)} \sum_{n = -N}^N \sum_{m = -M}^M \\
 && \phantom{+++++++++++} E\left(e^{\frac{2\pi}{\epsilon} i \left(\sum_{j = 1}^k Y_{-n-j}+ \sum_{j = 1}^k Y_{n+j} - \sum_{j = 1}^k Y_{-m-j}  - \sum_{j = 1}^k Y_{m+j} \right)}\right).
 \end{eqnarray*}
 The sum
 \begin{equation}\label{suminexponent}
 \sum_{j = 1}^k Y_{-n-j}+ \sum_{j = 1}^k Y_{n+j} - \sum_{j = 1}^k Y_{-m-j}  - \sum_{j = 1}^k Y_{m+j}
 \end{equation}
 may have cancelations among terms involving $n$ with terms involving $m.$
 Suppose that for a fixed $n$ and $m,$ there are $\tilde{k}_{m,n}$ indices that cancel in each of the 4 sums in (\ref{suminexponent}). Due to symmetry, the same number i.e., $\tilde{k}_{mn},$ of terms will cancel in each sum. Depending on $n$ and $m,$ $\tilde{k}_{mn}$ lies between $0$ and $k,$ i.e.,
 $0 \leq \tilde{k}_{mn} \leq k.$
 For the sake of making the notation less cumbersome, $\tilde{k}_{mn}$ will from now on be written as $\tilde{k}.$ When $m = n,$ $\tilde{k} = k.$ If $m > n+k$ or $n > m+k$ then $\tilde{k} =0.$ Each sum in (\ref{suminexponent}) has $k$ terms and  $\tilde{k}$ of these get cancelled leaving $(k - \tilde{k})$ terms. One can re-index the variables in (\ref{suminexponent}) and write it as
$$\sum_{j = 1}^k Y_{-n-j}+ \sum_{j = 1}^k Y_{n+j} - \sum_{j = 1}^k Y_{-m-j}  - \sum_{j = 1}^k Y_{m+j} = \pm Y_{\ell_1} \pm \cdots \pm Y_{\ell_{4(k - \tilde{k})}}$$
 where the sign depends on whether $m$ is less than or greater than $n.$
 Thus
\begin{eqnarray*}
E(|A_X[k]|^2) &=& \lim_{N \to \infty} \lim_{M \to \infty} \frac{1}{(2N+1)} \frac{1}{(2M+1)} \sum_{n = -N}^N \sum_{m = -M}^M E\left(e^{\frac{2\pi}{\epsilon}i \left(\pm Y_{\ell_1} \pm \cdots \pm Y_{\ell_{4(k - \tilde{k})}}\right)}\right)
\end{eqnarray*}
Due to the independence of the $Y_{\ell}$s, this means
\begin{eqnarray*}
E(|A_X[k]|^2) &=& \lim_{N \to \infty} \lim_{M \to \infty} \frac{1}{(2N+1)} \frac{1}{(2M+1)} \sum_{n = -N}^N \sum_{m = -M}^M \left[\phi_{Y_1}\left(\pm \frac{2\pi}{\epsilon}\right)\right]^{4(k - \tilde{k})} \\ 
&=& \lim_{N \to \infty} \lim_{M \to \infty} \frac{1}{(2N+1)} \frac{1}{(2M+1)} \sum_{n = -N}^N \sum_{m = -M}^M e^{-\frac{\sigma^2}{2}\left(\frac{2\pi}{\epsilon}\right)^2 4(k - \tilde{k})} \nonumber 
\end{eqnarray*}
The minimum is attained for $\tilde{k} = 0$ and the maximum at $\tilde{k} = k.$ Thus
\begin{eqnarray*}
E(|A_X[k]|^2) &\leq& \lim_{N \to \infty} \lim_{M \to \infty} \frac{1}{(2N+1)} \frac{1}{(2M+1)} \sum_{n = -N}^N \sum_{m = -M}^M 1 = 1 \quad \textrm{and} \\
E(|A_X[k]|^2) &\geq& \lim_{N \to \infty} \lim_{M \to \infty} \frac{1}{(2N+1)} \frac{1}{(2M+1)} \sum_{n = -N}^N \sum_{m = -M}^M
e^{-\frac{\sigma^2}{2}\left(\frac{2\pi}{\epsilon}\right)^2 4k} = e^{-\frac{\sigma^2}{2}\left(\frac{2\pi}{\epsilon}\right)^2 4k}.
\end{eqnarray*}
This gives
$$0 \leq V(A_X[k]) \leq 1 - e^{-2k\sigma^2\left(\frac{2\pi}{\epsilon}\right)^2 }.$$
A similar calculation can be done for $k <0.$ Thus for $k \neq 0,$
$$0 \leq V(A_X[k]) \leq 1 - e^{-2|k|\sigma^2\left(\frac{2\pi}{\epsilon}\right)^2 }.$$
 \subsection{Generalizing the construction to other random variables}
 So far the construction of discrete unimodular zero autocorrelation stochastic waveforms has been based on Gaussian random variables. This construction can be generalized to many other random variables. The unimodularity of the waveforms is not affected by using a different random variable.
 The following theorem characterizes the class of random variables that can be used to get the desired autocorrelation.
 \begin{theorem}\label{GeneralizationGaussian}
 Let $\{Y_{\ell}\}_{\ell \in \mathbb{Z}}$ be a sequence of i.i.d. random variables with characteristic function  $\phi_Y.$ Suppose that the probability density function of the $Y_{\ell}$s is even and that $\phi_Y(t)$ goes to $0$ as $t$ goes to infinity. Then, given $\epsilon,$ the waveform $X:\mathbb{Z} \to \mathbb{C}$ given by $$X[n] = e^{\frac{2\pi}{\epsilon} i \sum_{\ell=-n}^n Y_{\ell}}$$ has almost perfect autocorrelation.
 \end{theorem}
\begin{proof}
 Since the density function of each $Y_{\ell}$ is even this means that the characteristic function is real valued \cite{Karr93}. Following the calculation in the proof of Theorem \ref{DiscreteGaussianWaveforms}, the expected autocorrelation of $X$ for $k\neq 0$ is
 $$E(A_X[k]) = \left[\phi_{Y}\left(\frac{2\pi}{\epsilon}\right)\right]^{2|k|}$$
and this goes to zero with $\epsilon$ by the hypothesis.
\end{proof}
 \begin{example}\label{BilateralDist}
\rm
Suppose the $Y_{\ell}$s follow a \emph{bilateral} distribution that has density $e^{-|x|}$ with $x \in (-\infty, \infty)$ and characteristic function $\phi_Y(t) = \frac{1}{1+t^2}.$ Then for $k \neq 0,$
$$E(A_X[k]) = \left[ \frac{1}{1 + \left( \frac{2\pi}{\epsilon}\right)^2}\right]^{2|k|}$$ and this can be made arbitrarily small with $\epsilon.$

In the same way as was done in the Gaussian case, for $k > 0,$
\begin{eqnarray*}
E(|A_X[k]|^2) &=& \lim_{N \to \infty} \lim_{M \to \infty} \frac{1}{(2N+1)} \frac{1}{(2M+1)} \sum_{n = -N}^N \sum_{m = -M}^M
\left[\phi_{Y_1}\left(\pm \frac{2\pi}{\epsilon}\right)\right]^{4(k - \tilde{k})} \\
&\leq& 1
\qquad \textrm{and} \\
E(|A_X[k]|^2) &\geq& \left[\frac{1}{1 + \left(\frac{2\pi}{\epsilon}\right)^2}\right]^{4k}.
\end{eqnarray*} Thus
\begin{eqnarray*}
0 \leq V(A_X[k]) \leq 1 - \left(\frac{1}{1 + \left(\frac{2\pi}{\epsilon}\right)^2}\right)^{4|k|}.
\end{eqnarray*}
\end{example}
\begin{example}\label{CauchyDist}
\rm
Suppose that the $Y_{\ell}$s follow the \emph{Cauchy} distribution with density function $\frac{1}{\pi(1 + x^2)}.$ Note that, disregarding the constant $\pi,$ this is the characteristic function of the random variable considered in Example \ref{BilateralDist}. The characteristic function of the $Y_{\ell}$s is now $e^{-|t|},$ the same as the distribution function in Example \ref{BilateralDist}. For $k \neq 0,$
$$E(A_X[k]) = \left[\phi_{Y_1}\left(\frac{2\pi}{\epsilon}\right)\right]^{2|k|} = e^{-\frac{4\pi |k|}{\epsilon}}$$
which can be made arbitrarily small with $\epsilon.$ Also,
$$0 \leq V(A_X[k]) \leq 1 - e^{-\frac{8\pi |k|}{\epsilon}}.$$
\end{example}
\subsection{Higher dimensional case} \label{SectionHigherDim}
Here one is interested in constructing waveforms $v: \mathbb{Z} \to \mathbb{C}^d, $ $d \geq 2.$ It is desired that $v$ has unit norm and the expectation of its autocorrelation can be made arbitrarily small. One way to construct $v$ is based on the construction of the one dimensional example given in Section \ref{discrete_waveforms_Gaussian}. This is motivated by the higher dimensional construction in the deterministic case \cite{JBJD1}. As before, $\{Y_{\ell}\}_{\ell \in \mathbb{Z}}$ is a sequence of i.i.d. Gaussian random variables with mean zero and variance $\sigma^2.$ Next, one defines $X[n] = e^{\frac{2\pi}{\epsilon}i\sum_{\ell = -n}^n Y_{\ell}}.$ The waveform $v: \mathbb{Z} \to \mathbb{C}^d$ is then defined as
\begin{equation}\label{HigherDimSeq}
\forall m \in \mathbb{Z}, \quad v[m] = \frac{1}{\sqrt{d}}\left[\begin{array}{c}
X[m]\\
X[m+1] \\
\vdots
\\
X[m+d-1]
\end{array}\right].
\end{equation}
In this case, the autocorrelation is given by
\begin{equation} \label{HigherDimAutocor}
A_v[k] = \lim_{N \to \infty} \frac{1}{2N+1}\sum_{n = -N}^N \langle v[n + k], v[n]\rangle
\end{equation}
where $\langle . , . \rangle$ is the usual inner product in $\mathbb{C}^d.$ The length or norm of any $v[m]$ is thus given by $$\|v[m]\|^2 = \langle v[m], v[m]\rangle.$$ From (\ref{HigherDimSeq}),
\begin{eqnarray*}
\|v[m]\|^2 &=& \frac{1}{d}\sum_{n=0}^{d-1} X[m+n]\overline{X[m+n]} = \frac dd = 1.
\end{eqnarray*}
Thus the $v[m]$s are unit-normed. The following Theorem \ref{AutocorHighDim} shows that the expected autocorrelation of $v$ can be made arbitrarily small everywhere except at the origin.
\begin{theorem}\label{AutocorHighDim}
Given $\epsilon > 0,$ the waveform $v:\mathbb{Z} \to \mathbb{C}^d$ defined in (\ref{HigherDimSeq}) has autocorrelation $A_v$ such that
\begin{equation*}
E(A_v[k]) = \left\{ \begin{array}{ll}
1 & \textrm{if $k=0$} \\ e^{-|k|\sigma^2\left(\frac{2\pi}{\epsilon}\right)^2 } & \textrm{if $k \neq 0$}.
\end{array}
\right.
\end{equation*}
\end{theorem}
\begin{proof}
As defined in (\ref{HigherDimAutocor}),
$$A_v[k] = \lim_{N \to \infty} \frac{1}{2N+1}\sum_{n = -N}^N \langle v[n + k], v[n]\rangle.$$
When $k =0,$
$$
A_v[0] = \lim_{N \to \infty} \frac{1}{2N+1}\sum_{n = -N}^N \|v[n]\|^2 = 1.$$
Thus, $$E(A_v[0]) = 1.$$
For $k \neq 0,$ due to (\ref{HigherDimSeq}),
\begin{eqnarray*}
\langle v[n+k], v[n] \rangle &=& \frac 1d \left\langle \left[\begin{array}{c}
X[n+k]\\
X[n+k+1] \\
\vdots \\
X[n + k + d-1]
\end{array}\right], \left[\begin{array}{c}
X[n]\\
X[n+1] \\
\vdots \\
X[n + d-1]
\end{array}\right]\right\rangle \\
&=& \frac 1d \Big(X[n + k]\overline{X[n]} + X[n + k +1]\overline{X[n+1]} + \ldots + \\
&& \phantom{++++++++++} X[n + k+d-1]\overline{X[n+d-1]}\Big).
\end{eqnarray*}
Consider $k > 0.$
\begin{eqnarray*}
E(A_v[k]) &=& \lim_{N \to \infty} \frac{1}{2N+1}\sum_{n = -N}^N E(\langle v[n + k], v[n]\rangle) \\
&=& \lim_{N \to \infty} \frac{1}{2N+1}\sum_{n = -N}^N E\left(\frac 1d \sum_{m=0}^{d-1} X[n + k + m]\overline{X[n+m]}\right) \\
&=& \lim_{N \to \infty} \frac{1}{2N+1}\sum_{n = -N}^N \frac 1d \sum_{m=0}^{d-1} E\left(X[n + k + m]\overline{X[n+m]}\right) \\
&=& \lim_{N \to \infty} \frac{1}{2N+1}\sum_{n = -N}^N \frac 1d \sum_{m=0}^{d-1} E\left(e^{\frac{2\pi}{\epsilon}i(Y_{-(n+m+k)} + \ldots + Y_{-(n+m+1)}+Y_{n+m+1} + \ldots  + Y_{n+m+k})}\right) \\
&=& \lim_{N \to \infty} \frac{1}{2N+1}\sum_{n = -N}^N \frac 1d \sum_{m=0}^{d-1} E(e^{\frac{2\pi}{\epsilon}iY_1})^{2k} = e^{-\sigma^2k(\frac{2\pi}{\epsilon})^2}.
\end{eqnarray*}
Similarly, for $k < 0$ one gets
$$E(A_v[k]) =e^{\sigma^2k(\frac{2\pi}{\epsilon})^2}.$$
\end{proof}

\noindent
Thus the waveform $v$ as defined in this section is unit-normed and has autocorrelation that can be made arbitrarily small.
\begin{remark} \rm
As in the one dimensional construction, it is easy to see that here too the construction can be done with random variables other than the Gaussian. In fact, all random variables that can be used in the one dimensional case, i.e., ones satisfying the properties of Theorem \ref{GeneralizationGaussian}, can also be used for the higher dimensional construction.
\end{remark}
\subsection{Remark on the periodic case}
It can be shown that the periodic case follows the same nature as the aperiodic case. The sequence $X: \mathbb{Z}_n \to \mathbb{C}$ is defined in the same way as in Section \ref{DiscreteGaussianWaveforms}, i.e.,
$$\forall m \in \{0, 1, \ldots, n-1\}, \quad X[m] = e^{\frac{2\pi}{\epsilon}i\sum_{\ell = -m}^m Y_{\ell}}$$
where $Y_{\ell} \sim N(0, \sigma^2).$
Following the definition given in (\ref{periodicAutocorrelation}), when $k = 0,$ $$A_X[0] = \frac 1n \sum_{m=0}^{n-1} X[m]\overline{X[m]} = 1.$$
When $k \neq 0, $ the expectation of the autocorrelation is
$$
E(A_X[k]) = \frac 1n \sum_{m=0}^{n-1}E(X[m+k]\overline{X[m]})$$
For $k > 0,$
\begin{eqnarray*}
E(A_X[k])
&=& \frac{1}{n} \sum_{m = 0}^{n-1} E(e^{\frac{2\pi i}{\epsilon}\sum_{\ell = -m-k}^{m+k} Y_{\ell}}e^{-\frac{2\pi i}{\epsilon}\sum_{j = - m}^{m} Y_{j}}) \\
&=& \frac{1}{n} \sum_{m = 0}^{n-1} E(e^{\frac{2\pi i}{\epsilon}(\sum_{\ell = m+1}^{m+k} Y_{\ell} + \sum_{j = -m-k}^{-m-1} Y_{j})}) \\
&=& \frac{1}{n} \sum_{m = 0}^{n-1} \left[E\left(e^{\frac{2\pi i }{\epsilon}Y_1} \right)\right]^{2k} \\ &=& \left[ E\left(e^{\frac{2\pi i }{\epsilon}Y_1} \right) \right]^{2k}= \left[\phi_{Y_1}\left(\frac{2\pi}{\epsilon}\right)\right]^{2k} = e^{- k \sigma^2 \left(\frac{2\pi}{\epsilon}\right)^2 }
\end{eqnarray*}
where one uses the fact that the $Y_{\ell}$s are i.i.d.. A similar calculation for negative values of $k$ suggests that the autocorrelation can be made arbitrarily small, depending on $\epsilon,$ for all non-zero values of $k.$ Also, as in the aperiodic case, this result can be obtained for random variables other than the Gaussian.
\section{Construction of continuous stochastic waveforms}\label{stochasticCAZAC}
In this section continuous waveforms with almost perfect autocorrelation are constructed from a one dimensional Brownian motion.

For a continuous waveform $x: \mathbb{R} \to \mathbb{C},$ the autocorrelation $A_x:\mathbb{R} \to \mathbb{C}$ can be defined as
\begin{equation}
A_x(s) = \lim_{T \to \infty} \frac{1}{2T} \int_{-T}^T x(t + s) \overline{x(t)} \ \textrm{dt}.
\end{equation}
Let $\{W(t); \ t > 0\}$ be a one dimensional Brownian motion. Then $W(t)$ satisfies
\begin{itemize}
\item[(i)] $W(0) = 0$
\item[(ii)] $W(t+s) - W(s) \sim N(0, \sigma^2 t)$
\item[(iii)] $0< t_1 < \cdots < t_k,$
$ \ W(t_{i+1}) - W(t_i)$ are independent random variables.
\end{itemize}
\begin{theorem}\label{ContinuousWaveform}
Let $W(t)$ be the one dimensional Brownian motion and $\epsilon > 0$ be given.  Define $x : \mathbb{R} \to \mathbb{C}$ by
\begin{equation*}
x(t) = e^{\frac{2 \pi}{\epsilon}i W(t)} \quad \textrm{for $t \geq 0, $}
\end{equation*}
and $x(-t) = x(t).$ Then the autocorrelation of $x,$ $A_x,$ satisfies
\begin{equation*}
E(A_x(s)) = \left\{ \begin{array}{ll}
1 & \textrm{if $s=0$} \\ e^{-\frac{\sigma^2}{2}|s|\left(\frac{2\pi}{\epsilon}\right)^2} & \textrm{if $s \neq 0$}.
\end{array}
\right.
\end{equation*}
\end{theorem}
\begin{proof}
 We would like to evaluate
$$E(A_x(s)) = E\left(\lim_{T \to \infty} \frac{1}{2T} \int_{-T}^T x(t+s)\overline{x(t)} \ \textrm{dt} \right).$$
Let $s > 0$ and let $g_T(s) = \frac{1}{2T} \int_{-T}^T x(t+s)\overline{x(t)} \ \textrm{dt}.$
\begin{eqnarray*}
E(g_T) &=& \frac{1}{2T} \int_{-T}^T E(x(t+s)\overline{x(t)}) \ \textrm{dt} \\
&=& \frac{1}{2T} \int_{-T}^T E\left(e^{\frac{2\pi  }{\epsilon}i(W(t+s)-W(t))}\right) \ \textrm{dt} \\
&=& \frac{1}{2T}\int_{-T}^T \phi_{W(t+s)-W(t)}\left(\frac{2\pi}{\epsilon}\right) = e^{-\frac{\sigma^2}{2}s\left(\frac{2\pi}{\epsilon}\right)^2} < \infty.
\end{eqnarray*}
Thus each $g_T$ is integrable and further $|g_T| \leq 1.$ Let $h(t) = 1; \ t \in \mathbb{R}.$ Then $E(h) = 1.$ Therefore, by the Dominated Convergence Theorem, and properties of Brownian motion and characteristic functions, one gets
\begin{eqnarray*}
E(A_x(s)) &=& E\left(\lim_{T \to \infty} \frac{1}{2T} \int_{-T}^T x(t+s)\overline{x(t)} \ \textrm{dt} \right) \\
&=& \lim_{T \to \infty} \frac{1}{2T} \int_{-T}^T E\left(e^{\frac{2 \pi }{\epsilon}i (W(t+s) - W(t))}\right)  \ \textrm{dt}  \\
&=& \lim_{T \to \infty}\frac{1}{2T}\int_{-T}^T \phi_{W(t+s)-W(t)}\left(\frac{2\pi}{\epsilon}\right) = e^{-\frac{\sigma^2}{2}s\left(\frac{2\pi}{\epsilon}\right)^2}
\end{eqnarray*}
which can be made arbitrarily small based on $\epsilon.$ Similarly, $$E(A(-s)) = \lim_{T \to \infty}\frac{1}{2T}\int_{-T}^T \phi_{W(t)-W(t-s)}\left(-\frac{2\pi}{\epsilon}\right) = e^{-\frac{\sigma^2}{2}s\left(-\frac{2\pi}{\epsilon}\right)^2}  = e^{-\frac{\sigma^2}{2}s\left(\frac{2\pi}{\epsilon}\right)^2}.$$ \end{proof}
\section{Connection to frames}\label{StochasticFrames}
Consider the mapping $v: \mathbb{Z} \to \mathbb{C}^d$ given by
\begin{equation}\label{stochastic_frame}
v(k) = \frac{1}{\sqrt{d}}\left(\begin{array}{c}
X[k] \\
X[k+1] \\
\vdots \\
X[k+d-1]
\end{array}\right)
\end{equation}
where $X[k] = e^{\frac{2\pi}{\epsilon}i \sum_{\ell = -k}^{k} Y_{\ell}},$ as defined in Section \ref{discrete_waveforms_Gaussian}.

Let $M \geq d$ and consider the set $V = \{v(1), v(2), \ldots, v(M)\}$ of $M$ unit vectors in $\mathbb{C}^d.$ The matrix
\begin{equation*}
F = \frac{1}{\sqrt{d}}\left[
\begin{array}{cccc}
\overline{X[1]} & \overline{X[2]} & \cdots &  \overline{X[d]} \\
\overline{X[2]} & \overline{X[3]} & \cdots &  \overline{X[d+1]} \\
\vdots & \vdots & \cdots &  \vdots\\
\overline{X[M]} & \overline{X[M+1]} & \cdots &  \overline{X[M+d-1]} \\
\end{array}
\right].
\end{equation*}
 is the matrix of the analysis operator corresponding to $V.$
The frame operator of $V$ is $\mathcal{F} = F^*F,$ i.e.,
\begin{equation*}
\mathcal{F} = \frac{1}{d}
\left[
\begin{array}{cccc}
X[1] & X[2] & \cdots & X[M] \\
X[2] & X[3] & \cdots  & X[M+1] \\
\vdots & \vdots & \cdots & \vdots \\
X[d] & X[d+1] & \cdots & X[M+d-1] \\
\end{array}
\right] \left[
\begin{array}{cccc}
\overline{X[1]} & \overline{X[2]} & \cdots &  \overline{X[d]} \\
\overline{X[2]} & \overline{X[3]} & \cdots &  \overline{X[d+1]} \\
\vdots & \vdots & \cdots &  \vdots\\
\overline{X[M]} & \overline{X[M+1]} & \cdots &  \overline{X[M+d-1]} \\
\end{array}
\right].
\end{equation*}
The entries of $\mathcal{F}$ are given by $\mathcal{F}_{m,m} = \frac Md$ and for $m \neq n, \ m> n,$
\begin{eqnarray*}
 \mathcal{F}_{m,n} &=& \frac 1d\left( \begin{array}{cccc}
X[m] & X[m+1] & \cdots & X[M+m-1]
\end{array}
\right)
\left( \begin{array}{c}
\overline{X[n]} \\
\overline{X[n+1]}  \\
\vdots \\
 \overline{X[M+n-1]}
\end{array}
\right) \\
&=& \frac Md \frac 1M \sum_{\ell=0}^{M-1}X[m+\ell]\overline{X[n+\ell]} \\
&=& \frac Md\left(\frac{1}{M}\sum_{\ell=0}^{M-1} e^{\frac{2\pi}{\epsilon}i\left(Y_{-m - \ell} + \cdots + Y_{-n- \ell - 1} +  Y_{n + \ell + 1} + \cdots + Y_{m + \ell}\right)}\right).
\end{eqnarray*}
Note that since $\mathcal{F}$ is self-adjoint, $\mathcal{F}_{m,n} = \overline{\mathcal{F}}_{n, m}.$
It is desired that $V$ emulates a tight frame, i.e, $\mathcal{F}$ is close to a constant times the identity, in this case, $\frac Md$ times the identity. Alternatively, it is desirable that the eigenvalues of $\mathcal{F}$ are all close to each other and close to $\frac Md.$ In this case, due to the stochastic nature of the frame operator, one studies the expectation of the eigenvalues of $\mathcal{F}.$
%
\subsection{Frames in $\bm{\mathbb{C}^2}$}
This section discusses the construction of sets of vectors in $\mathbb{C}^2$ as given by (\ref{stochastic_frame}). The frame properties of such sets are analyzed. In fact, it is shown that the expectation of the eigenvalues of the frame operator are close to each other, the closeness increasing with the size of the set. The bounds on the probability of deviation of the eigenvalues from the expected value is also derived. The related inequalities arise from an application of Theorem \ref{AzumaInequality} \cite{Hoeffding} below.
\begin{theorem}[Azuma's Inequality \cite{Hoeffding}]\label{AzumaInequality}
Suppose that $\{X_k \ : \ k = 0, 1,2, \ldots\}$ is a martingale and
$$|X_k - X_{k-1}| < c_k,$$ almost surely. Then for all positive integers $N$ and all positive reals $t,$
$$P(|X_N - X_0|\geq t) \leq 2e^{\left(\frac{-t^2}{2\sum_{k=1}^N c_k^2}\right)}.$$
\end{theorem}
Consider  $M \geq 3$ vectors in $\mathbb{C}^2,$ i.e., $d = 2$ in (\ref{stochastic_frame}).
Then $v : \mathbb{Z} \to \mathbb{C}^2$ and
\begin{equation}\label{frameC2}
v(1) = \frac{1}{\sqrt{2}}\left(\begin{array}{c}
X[1] \\
X[2]
\end{array}\right), \ v(2) = \frac{1}{\sqrt{2}}\left(\begin{array}{c}
X[2] \\
X[3]
\end{array}\right), \ldots,  v(M) = \frac{1}{\sqrt{2}}\left(\begin{array}{c}
X[M] \\
X[M+1]
\end{array}\right).
\end{equation}
Considering the set $V = \{v(1), v(2), \ldots, v(M)\},$ the frame operator of $V$ is
\begin{eqnarray}
\mathcal{F} &=& \frac 12 \left[\begin{array} {ccccc}
X[1] & X[2] & X[3] & \cdots & X[M]\\
X[2] & X[3] & X[4] & \cdots & X[M+1]
\end{array}\right]\left[\begin{array} {cc}
\overline{X[1]} & \overline{X[2]}  \\
\overline{X[2]} & \overline{X[3]} \\
\overline{ X[3]} & \overline{X[4]} \\
\vdots & \vdots \\
\overline{ X[M]} & \overline{X[M+1]}
\end{array}\right] \nonumber \\
\textrm{or, } \quad \mathcal{F}&=& \frac M2 \left[
\begin{array}{cc}
1 & \frac 1M \sum_{m = 1}^M X[m] \overline{X[m+1]} \\
\frac 1M \sum_{m = 1}^M \overline{X[m]} X[m+1] & 1
\end{array}
\right]. \label{frameoperatorC2}
\end{eqnarray}
%
\begin{theorem}\label{EValuesFrame}
(a) Consider the set $V = \{v(1), v(2), \ldots, v(M)\} \subseteq \mathbb{C}^2,$ $M \geq 3,$ where the vectors $v(n)$ are given by (\ref{frameC2}).
The minimum eigenvalue, $\lambda_{\min}(\mathcal{F}),$ and the maximum eigenvalue, $\lambda_{\max}(\mathcal{F}),$ of the frame operator of  $V$ satisfy
\begin{equation}\label{eigenvalueslimitsMge2}
\frac M2\left(1 - \delta \right) \leq E(\lambda_{\min}(\mathcal{F}))
\leq E(\lambda_{\max}(\mathcal{F})) \leq \frac M2\left(1 + \delta \right)
\end{equation}
where $\delta = \sqrt{\frac 1M + \frac{M-1}{M} e^{-2\sigma^2 \left(\frac{2\pi}{\epsilon}\right)^2}}.$ \\
(b) The deviation of the minimum and maximum eigenvalue of $\mathcal{F}$ from their expected value is given, for all positive reals $r,$ by
\begin{eqnarray*}
& P(|\lambda_{\min}(\mathcal{F}) - E(\lambda_{\min}(\mathcal{F}))| > r) \leq 2 e^{-\frac{4r^2}{8M^3}}, \\
& P(|\lambda_{\max}(\mathcal{F}) - E(\lambda_{\max}(\mathcal{F}))| > r) \leq 2 e^{-\frac{4r^2}{8M^3}}.
\end{eqnarray*}
\end{theorem}
\begin{proof}
(a)
The frame operator of $V=\{v(1), v(2), \ldots, v(M)\}$ is
given in (\ref{frameoperatorC2}).
The eigenvalues of $\frac 2M \mathcal{F}$ are $\lambda_{1} = 1 - |\alpha|$ and $\lambda_{2} = 1 + |\alpha|$ where
$$\alpha = \frac 1M \sum_{m = 1}^M X[m] \overline{X[m+1]}.$$ Let
\begin{eqnarray}
\gamma_1 &=& X[1]\overline{X[2]} = e^{-\frac{2\pi}{\epsilon}i(Y_{-2} + Y_2)}, \nonumber\\
\gamma_2 &=& X[2]\overline{X[3]} = e^{-\frac{2\pi}{\epsilon}i(Y_{-3} + Y_3)}, \nonumber \\
\vdots && \nonumber \\
\gamma_M &=& X[M]\overline{X[M+1]} =  e^{-\frac{2\pi}{\epsilon}i(Y_{-(M+1)} + Y_{(M+1)})}, \nonumber
\end{eqnarray}
so that $$\alpha = \frac{\gamma_1 + \gamma_2 + \cdots + \gamma_M}{M}.$$
Note that for $m \neq n,$ $\gamma_m$ and $\gamma_n$ are independent and so
$E(\gamma_m \overline{\gamma_n}) = E(\gamma_m)E(\overline{\gamma_n}).$
Also, since the $Y_{\ell}$s are i.i.d. and the characteristic function of the $Y_{\ell}$s is symmetric,  $$\forall \ 1 \leq m \leq M, \quad E(\gamma_m) = E(e^{\frac{2\pi}{\epsilon}i(Y_{-(m+1)} + Y_{m+1})}) = \{E(e^{\frac{2\pi}{\epsilon}iY_1})\}^2 = e^{-\sigma^2(\frac{2\pi}{\epsilon})^2} = E(\overline{\gamma}_m)$$ and therefore
$$E(\gamma_{m}\overline{\gamma_n}) = e^{-2\sigma^2(\frac{2\pi}{\epsilon})^2}.$$
Thus
\begin{eqnarray*}
E(|\alpha|^2) &=& E(\alpha\overline{\alpha}) = \frac{1}{M^2} E((\gamma_1 + \gamma_2 + \cdots + \gamma_M)(\overline{\gamma_1} + \overline{\gamma_2} + \cdots + \overline{\gamma_M})) \\
&=& \frac{1}{M^2}E(|\gamma_1|^2 + |\gamma_2|^2 + \cdots + |\gamma_M|^2 + \sum_{m \neq n}\gamma_{m}\overline{\gamma_n}) \\
&=& \frac 1M + \frac{1}{M^2} E\left(\sum_{m \neq n}\gamma_{m}\overline{\gamma_n}\right)\\
&=& \frac 1M + \frac{1}{M^2} \sum_{m \neq n} E(\gamma_{m}\overline{\gamma_n}) \\
&=& \frac1M + \frac{M-1}{M} e^{-2\sigma^2(\frac{2\pi}{\epsilon})^2}.
\end{eqnarray*}
The above estimate on $E(|\alpha|^2)$ implies that
\begin{equation}\label{estimate_of_alpha}
E(|\alpha|) \leq \sqrt{E(|\alpha|^2)} = \sqrt{\frac1M + \frac{M-1}{M}e^{-2\sigma^2(\frac{2\pi}{\epsilon})^2}}.
\end{equation}
Since $E(\lambda_1) = 1 - E(|\alpha|)$ and $E(\lambda_2) = 1 + E(|\alpha|),$ (\ref{estimate_of_alpha}) implies
\begin{equation*}
1 - \sqrt{\frac 1M + \frac{M-1}{M} e^{-2\sigma^2 \left(\frac{2\pi}{\epsilon}\right)^2}} \leq E(\lambda_{1})
\leq E(\lambda_{2}) \leq 1 + \sqrt{\frac 1M + \frac{M-1}{M} e^{-2\sigma^2 \left(\frac{2\pi}{\epsilon}\right)^2}}.
\end{equation*}
Noting that $\lambda_{\min}(\mathcal{F}) = \frac M2 \lambda_1$ and $\lambda_{\max}(\mathcal{F}) = \frac M2 \lambda_2,$ one finally gets, after
setting $\delta = \sqrt{\frac 1M + \frac{M-1}{M} e^{-2\sigma^2 \left(\frac{2\pi}{\epsilon}\right)^2}},$
\begin{equation*}
\frac M2\left(1 - \delta \right) \leq E(\lambda_{\min}(\mathcal{F}))
\leq E(\lambda_{\max}(\mathcal{F})) \leq \frac M2\left(1 + \delta \right).
\end{equation*}
(b) To prove  (b) we use the Doob martingale and Azuma's inequality \cite{Hoeffding}. For $n = 2, \ldots, M+1,$ let $Z_{n-1} = Y_{-n} + Y_n.$ Here the Doob martingale is the sequence $\{U_0, U_1, \ldots, U_{M-1}\}$ where
$$\textrm{for $k = 1, \ldots, M-1,$} \quad U_k = E\left(\frac 1M \left|\sum_{j=1}^M e^{-\frac{2\pi}{\epsilon}i Z_j}\right| \ | Z_1, Z_2, \ldots, Z_k\right)$$ and
$$U_0 = E\left(\frac 1M \left|\sum_{j=1}^M e^{-\frac{2\pi}{\epsilon}i Z_j}\right| \right).$$
Note that $U_0 = E(|\alpha|)$ and $U_{M-1} = |\alpha|.$ Also,
$$|U_k - U_{k-1}| \leq |U_k|+|U_{k-1}| \leq 2.$$ So by Azuma's Inequality (see Theorem \ref{AzumaInequality})
$$P(\left|U_{M-1}-U_0\right|\geq r)=P(\left| |\alpha|- E(|\alpha|)\right|\geq r) \leq 2 e^{-\frac{r^2}{2\sum_{k=1}^M 2^2}} = 2 e^{-\frac{r^2}{8M}}.$$
Since $|\lambda_1 - E(\lambda_1)| = |\lambda_2 - E(\lambda_2)| = ||\alpha|- E(|\alpha|)|,$ this means
$$P(|\lambda_1 - E(\lambda_1)|>r)\leq 2 e^{-\frac{r^2}{8M}}$$ and
$$P(|\lambda_2 - E(\lambda_2)|>r)\leq 2 e^{-\frac{r^2}{8M}}.$$
Going back to the actual frame operator $\mathcal{F},$ whose eigenvalues are $\frac M2 \lambda_1$ and $\frac M2 \lambda_2,$ the following estimates hold.
\begin{eqnarray*}
P(|\lambda_{\max}(\mathcal{F}) - E(\lambda_{\max}(\mathcal{F}))|>r) &=& P\left(\frac M2\left||\alpha|- E(|\alpha|)\right|>r\right) \\
&=& P\left(\left||\alpha|- E(|\alpha|)\right|>\frac 2Mr\right) \leq 2 e^{-\frac{4r^2}{8M^3}}
\end{eqnarray*}
and
\begin{eqnarray*}
P(|\lambda_{\min}(\mathcal{F}) - E(\lambda_{\min}(\mathcal{F}))|>r) &=& P\left(\frac M2\left||\alpha|- E(|\alpha|)\right|>r\right) \\
&=& P\left(\left||\alpha|- E(|\alpha|)\right|>\frac 2Mr\right) \leq 2 e^{-\frac{4r^2}{8M^3}}.
\end{eqnarray*}
\end{proof}
\begin{corollary} The eigenvalues of the frame operator considered
in Theorem \ref{EValuesFrame} satisfy, for all positive reals $r,$
\begin{eqnarray*}
&& P\left(\lambda_{\min}(\mathcal{F}) < \frac M2\left(1 - \delta\right) - r\right) \leq e^{-4r^2/8M^3}, \\
&& P\left(\lambda_{\max}(\mathcal{F}) > \frac M2\left(1 + \delta\right) + r\right) \leq e^{-4r^2/8M^3}
\end{eqnarray*}
where
$\delta = \sqrt{\frac 1M + \frac{M-1}{M} e^{-2\sigma^2 \left(\frac{2\pi}{\epsilon}\right)^2}}.$
\end{corollary}
\begin{proof}
Due to part (a) of Theorem \ref{EValuesFrame}
$$\lambda_{\min}(\mathcal{F}) < \frac M2(1 - \delta) - r \implies \lambda_{\min}(\mathcal{F}) < E(\lambda_{\min}(\mathcal{F})) - r.$$ This implies, as a consequence of part (b) of Theorem \ref{EValuesFrame}, that
$$P\left(\lambda_{\min}(\mathcal{F}) < \frac M2(1 - \delta) - r \right) \leq P\left(\lambda_{\min}(\mathcal{F}) < E(\lambda_{\min}(\mathcal{F})) - r \right) \leq e^{-4r^2/8M^3}.$$
In a similar way, from part (a) of Theorem \ref{EValuesFrame}
$$\lambda_{\max}(\mathcal{F}) > \frac M2(1 + \delta) + r \implies \lambda_{\max}(\mathcal{F}) > E(\lambda_{\max}(\mathcal{F})) + r$$
which implies, as a consequence of part (b) of Theorem \ref{EValuesFrame}, that
$$P\left(\lambda_{\max}(\mathcal{F}) > \frac M2(1 + \delta) + r \right) \leq P\left(\lambda_{\max}(\mathcal{F}) > E(\lambda_{\max}(\mathcal{F})) + r \right) \leq e^{-4r^2/8M^3}.$$
\end{proof}
\begin{remark}
\rm
In Theorem \ref{EValuesFrame}, as $M$ tends to infinity, the value of $\delta$ in (\ref{eigenvalueslimitsMge2}) can be made arbitrarily small based on the choice of $\epsilon.$ This in turn implies that the two eigenvalues can be made arbitrarily close to each other with $\epsilon.$
On the other hand, for a fixed $M,$ as $\epsilon$ tends to zero, (\ref{eigenvalueslimitsMge2}) becomes
$$\frac M2\left(1 - \sqrt{\frac 1M} \right) \leq E(\lambda_{\min}(\mathcal{F}))
\leq E(\lambda_{\max}(\mathcal{F})) \leq \frac M2\left(1 + \sqrt{\frac 1M}\right).$$
\end{remark}
\subsection{Frames in $\bm{\mathbb{C}^d;}$ $\bm{d > 2}$}
For general $d$ and $M,$ in order to use existing results on the concentration of eigenvalues of random matrices \cite{Bai99,Ledoux2001}, a slightly different construction of the frame needs to be considered. Let $\{Y_{mn}\}_{m,n \in \mathbb{Z}}$ be i.i.d.~ random variables following a Gaussian distribution with mean zero and variance $\sigma^2.$
It can be shown that $$E(e^{\frac{2\pi}{\epsilon} i Y_{mn}}) = e^{-\frac{\sigma^2}{2}\left(\frac{2\pi}{\epsilon}\right)^2}$$
and the variance
$$V(e^{\frac{2\pi}{\epsilon} i Y_{mn}}) = 1 - e^{-\sigma^2\left(\frac{2\pi}{\epsilon}\right)^2}.$$
 One can define the following two dimensional sequence. For $m, n \in \mathbb{Z},$
$$X_{mn} = e^{\frac{2\pi}{\epsilon} i Y_{mn}} - e^{-\frac{\sigma^2}{2}\left(\frac{2\pi}{\epsilon}\right)^2}.$$
Consider the mapping $v:\mathbb{Z} \to \mathbb{C}^d$ given by
\begin{equation}\label{generalframe}
v(\ell) = \frac{1}{\sqrt{d}}\left(\begin{array}{c}
X_{1\ell} \\
X_{2\ell} \\
\vdots\\
X_{d\ell}
\end{array}\right).
\end{equation}
As before, let $M \geq d$ and consider the set of $M$ unit vectors $V = \{v(1), v(2), \ldots, v(M)\}$ in $\mathbb{C}^d.$ The frame operator of this set is
\begin{displaymath}
\mathcal{F} = \frac{1}{d}
\left[
\begin{array}{cccc}
X_{11} & X_{12} & \cdots & X_{1M} \\
X_{21} & X_{22} & \cdots  & X_{2M} \\
\vdots & \vdots & \cdots & \vdots \\
X_{d1} & X_{d2} & \cdots & X_{dM} \\
\end{array}
\right] \left[
\begin{array}{cccc}
\overline{X}_{11} & \overline{X}_{21} & \cdots  & \overline{X}_{d1} \\
\overline{X}_{12} & \overline{X}_{22} & \cdots & \overline{X}_{d2} \\
\vdots & \vdots & \cdots & \vdots \\
\overline{X}_{1M} & \overline{X}_{2M} & \cdots & \overline{X}_{dM} \\
\end{array}
\right].
\end{displaymath}
Let
\begin{equation} \label{PreFrameOperator}
A = \frac{1}{\sqrt{d}}
\left[
\begin{array}{cccc}
X_{11} & X_{12} & \cdots & X_{1M} \\
X_{21} & X_{22} & \cdots  & X_{2M} \\
\vdots & \vdots & \cdots & \vdots \\
X_{d1} & X_{d2} & \cdots & X_{dM}
\end{array}
\right]
\end{equation}
so that $\mathcal{F} = AA^*.$ The matrix $A$ has entries with mean zero and variance $\hat{\sigma}^2 = \frac 1d(1 - e^{-\sigma^2\left(\frac{2\pi}{\epsilon}\right)^2}).$ According to results in \cite{Bai99}, if $\frac dM \to c $ as $d, M \to \infty,$ then the smallest and largest eigenvalues of $\mathcal{F}$ converge almost surely to $\hat{\sigma}^2(1 - \sqrt{c})^2$ and $\hat{\sigma}^2(1 + \sqrt{c})^2$ respectively.
\begin{theorem}\label{EValueDeviation}
Let $s_1(A) \leq s_2(A) \leq \ldots \leq s_d(A)$ be the singular values of the matrix $A$ given by (\ref{PreFrameOperator}). Then the following hold.
\\
(a) Given $\epsilon_0,$ there is a large enough $d$ such that
\begin{equation}\label{largestsingularvalue}
P\left(s_d(A) \geq \hat{\sigma}\left(1 + \sqrt{\frac dM}\right) + \epsilon_0 + r\right) \leq 2 e^{-r^2 d/16}.
\end{equation}
(b)
\begin{equation}\label{smallestsingularvalue}
P(s_1(A) \leq c_1) \leq e^{-c_2 M}
\end{equation}
where $c_1$ and $c_2$ are universal positive constants.
\end{theorem}
\begin{proof}
Let $s_d$ be the mapping that associates to a matrix $A$ it largest singular value. Equip $\mathbb{C}^{dM}$ with the Frobenius norm
$$\|A\|^2 := \textrm{Tr}(AA^*) = \sum_{m,n}|A_{mn}|^2.$$ Then the mapping $s_d$ is convex and $1$-Lipschitz in the sense that 
$$|s_d(A) - s_d(A')|\leq \|A - A'\|$$ for all pairs $(A, A')$ of $d$ by $M$ matrices \cite{Ledoux2001}.

We think of $A$ as a random vector in $\mathbb{R}^{2dM}.$ The real and imaginary parts of the entries of $\frac{1}{\sqrt{d}}A$ are supported in $[-\frac{1}{\sqrt{d}}, \frac{1}{\sqrt{d}}].$
Let $P$ be a product measure on $[-\frac{1}{\sqrt{d}}, \frac{1}{\sqrt{d}}]^{2dM}$ Then as consequence of concentration inequality (Corollary $4.10$, \cite{Ledoux2001}) we have
$$P(|s_d(A) - m(s_d)| \geq r) \leq 4 e^{-r^2 d/16}$$ where $m(s_d)$ is the median of $s_d(A).$ It is known that the minimum and maximum singular values of $A$ converge almost surely to $\hat{\sigma}(1 - \sqrt{c})$ and $\hat{\sigma}(1 + \sqrt{c})$ respectively as $d, M $ tend to infinity and $\frac{d}{M} \to c$. As a consequence, for each $\epsilon_0$ and $M$ sufficiently large, one can show that the medians belong to the fixed interval $[\hat{\sigma}(1 - \sqrt{\frac dM}) - \epsilon_0,\hat{\sigma}(1 + \sqrt{\frac dM}) + \epsilon_0]$ which gives
$$P\left(s_d(A) \geq \hat{\sigma}\left(1 + \sqrt{\frac dM}\right) + \epsilon_0 + r\right) \leq 2 e^{-r^2 d/16}.$$

For the smallest singular value we cannot use the concentration inequality as used for $s_d$ since the smallest singular value is not convex. However, following results in \cite{LPRTJ2004} (Theorem 3.1) that have been used in \cite{CanTao06} in a similar situation as here, one can say that whenever $M > (1+\delta)d$ where $\delta$ is greater than a small constant,
\begin{equation*} 
P(s_1(A) \leq c_1) \leq e^{-c_2M}
\end{equation*}
where $c_1$ and $c_2$ are positive universal constants.
\end{proof}
\begin{remark}
\rm
Note that the square of the singular values of $A$ are the eigenvalues of $\mathcal{F}$ and so the estimates given in (\ref{largestsingularvalue})-(\ref{smallestsingularvalue}) give insight into the corresponding deviation of the eigenvalues of the frame operator $\mathcal{F}.$
\end{remark}
\begin{remark}[\emph{Connection to compressed sensing}] \rm
The theory of compressed sensing \cite{Candes_06,CRT06, Don06} states that it is possible to recover a sparse signal from a small number of measurements. A signal $x\in \mathbb{C}^M$ is $k$-sparse in a basis $\Psi=\{\psi_j\}_{j=1}^M$ if $x$ is a weighted superposition of at most $k$ elements of $\Psi$. Compressed sensing broadly refers to the inverse problem of reconstructing such a signal $x$ from linear measurements
$\{y_\ell=\langle x,\phi_\ell\rangle \ | \ \ell=1,\ldots, d\}$ with $d<M$, ideally with $d\ll M$. In the general setting, one has $\Phi^* x = y$, where $\Phi$ is a $d \times M$ sensing matrix having the measurement vectors $\phi_\ell$ as its columns, $x$ is a length-$M$ signal and $y$ is a length-$d$ measurement.

The standard compressed sensing technique guarantees exact recovery of the original signal with very high probability if the sensing matrix satisfies the Restricted Isometry Property (RIP). This means that for a fixed $k$, there exists a small number $\delta_k$, such that
\begin{equation*}
 (1-\delta_k)\|x\|^2_{\ell_2}\leq\|\Phi x\|^2_{\ell_2}\leq(1+\delta_k)\|x\|^2_{\ell_2},
\end{equation*}
for any $k$-sparse signal $x$.
By imitating the work done in \cite{CanTao06} (Lemmas 4.1 and 4.2), it can be shown, due to Theorem \ref{EValueDeviation}, that matrices $A$ of the type given in (\ref{PreFrameOperator}) satisfy the RIP condition and can therefore be used as measurement matrices in compressed sensing. These matrices are different from the traditional random matrices used in compressed sensing in that their entries are complex-valued unimodular instead of  being real-valued and not unimodular.
\end{remark}
\begin{example}
\rm
This example illustrates the ideas in this subsection.
First consider $M = 5$ and $d=3$ so that there are $5$ vectors in $\mathbb{C}^3.$ Taking from a normal distribution with mean $0$ and variance $\sigma = 1,$  a realization of the matrix $[Y_{mn}]_{1\leq m \leq 3, 1\leq n \leq 5}$ is
\begin{displaymath}
\left[\begin{array}{ccccc}
-0.0353  &  0.5004 &  -0.6299 &  -0.1472  &  0.4003 \\
   -0.4804 &  -0.9344 &   0.4220 &  -0.9509  &  0.2783 \\
   -0.8609  & -0.4822  & -0.4680 &  -0.0913 &   1.2284 \\
\end{array}\right].
\end{displaymath}
Then taking $\epsilon = 0.001,$ $A = \frac{1}{\sqrt{3}}[e^{\frac{2\pi}{\epsilon}iY_{mn}}]$ is
\begin{displaymath}
A = \frac{1}{\sqrt{3}}\left[\begin{array}{ccccc}
-0.27 - 0.96i & -0.89 + 0.44i &  0.85 + 0.52i &   0.33 - 0.94i &  -0.24 + 0.97i \\
  -0.92 - 0.39i & -0.89 - 0.46i &  0.99 + 0.05i &   0.93 + 0.37i & -0.47 + 0.88i \\
   0.74 + 0.68i &  0.16 - 0.99i &  0.99 + 0.09i & -0.30 - 0.95i & -0.74 + 0.67i
\end{array}\right].
\end{displaymath}
The condition number, ratio of the maximum and minimum eigenvalues, of $\mathcal{F} = 4.8667.$ As the number of vectors $M$ is increased, the condition number gets closer to $1.$ Figure \ref{fig:CondNoFrameMatrix} shows the behavior of the condition number with the increase in the number of vectors. 
\begin{figure}[t]
\begin{center}
\includegraphics[width=4in]{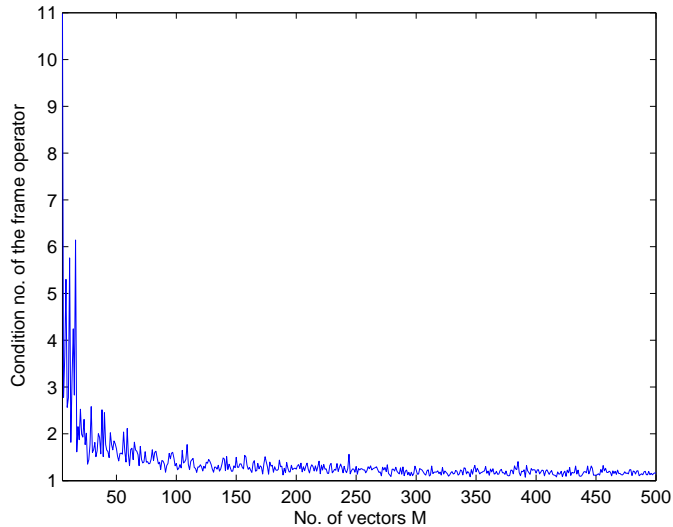}
\caption{Behavior of the condition number of the frame operator with increasing size of the frame; $\epsilon = 0.0001, d = 3, \sigma = 1$}
\label{fig:CondNoFrameMatrix}
\end{center}
\end{figure}
\end{example}
%
\section{Conclusions}
The construction of discrete unimodular stochastic waveforms with arbitrarily small expected autocorrelation has been proposed. This is motivated by the usefulness of such waveforms in the areas of radar and communications. The family of random variables that can be used for this purpose has been characterized. Such construction been done in one dimension and generalized to higher dimensions. Further, such waveforms have been used to construct frames in $\mathbb{C}^d$ and the frame properties of such frames have been studied. Using Brownian motion, this idea is also extended to the construction of continuous unimodular stochastic waveforms whose autocorrelation can be made arbitrarily small in expectation.
\section*{Acknowledgments}
The author wishes to acknowledge support from AFOSR Grant No. FA9550-10-1-0441 for conducting this research. The author is also grateful to Frank Gao and Ross Richardson for their generous help with probability theory.
%


\bibliographystyle{elsarticle-num}
\bibliography{StochasticWaveforms_refs}

\begin{thebibliography}{10}
\expandafter\ifx\csname url\endcsname\relax
  \def\url#1{\texttt{#1}}\fi
\expandafter\ifx\csname urlprefix\endcsname\relax\def\urlprefix{URL }\fi
\expandafter\ifx\csname href\endcsname\relax
  \def\href#1#2{#2} \def\path#1{#1}\fi

\bibitem{AB1}
L.~Auslander, P.~E. Barbano, Communication codes and {B}ernoulli
  transformations, Appl. Comput. Harmon. Anal. 5~(2) (1998) 109--128.

\bibitem{JBJD1}
J.~J. Benedetto, J.~J. Donatelli, Ambiguity function and frame theoretic
  properties of periodic zero autocorrelation waveforms, IEEE J. Special Topics
  Signal Processing 1 (2007) 6--20.

\bibitem{HK1}
T.~Helleseth, P.~V. Kumar, Sequences with low correlation, in: Handbook of
  coding theory, Vol. I, II, North-Holland, Amsterdam, 1998, pp. 1765--1853.

\bibitem{LM1}
N.~Levanon, E.~Mozeson, Radar {S}ignals, Wiley Interscience, IEEE Press, 2004.

\bibitem{MLL1}
M.~L. Long, Radar {R}eflectivity of {L}and and {S}ea, Artech House, 2001.

\bibitem{MWH1}
W.~H. Mow, A new unified construction of perfect root-of-unity sequences, in:
  Proc. IEEE 4th International Symposium on Spread Spectrum Techniques and
  Applications (Germany), 1996, pp. 955--959.

\bibitem{FEN1}
F.~E. Nathanson, Radar {D}esign {P}rinciples - {S}ignal {P}rocessing and the
  {E}nvironment, SciTech Publishing Inc., Mendham, NJ, 1999.

\bibitem{GWS1}
G.~W. Stimson, Introduction to {A}irborne {R}adar, SciTech Publishing Inc.,
  Mendham, NJ, 1998.

\bibitem{UY1}
S.~Ulukus, R.~D. Yates, Iterative construction of optimum signature sequence
  sets in synchronous {CDMA} systems, IEEE Trans. Inform. Theory 47~(5) (2001)
  1989--1998.

\bibitem{SV1}
S.~Verd{\'u}, Multiuser {D}etection, Cambridge University Press, Cambridge, UK,
  1998.

\bibitem{BD1}
J.~J. Benedetto, S.~Datta, Construction of infinite unimodular sequences with
  zero autocorrelation, Advances in Computational Mathematics 32~(2) (2010) 191
  -- 207.

\bibitem{Cochran05}
D.~Cochran, Waveform-agile sensing: opportunities and challenges, in: IEEE
  International Conference on Acoustics, Speech, and Signal Processing, Vol.~5,
  2005, pp. 877 -- 880.
\newblock \href {http://dx.doi.org/10.1109/ICASSP.2005.1416444}
  {\path{doi:10.1109/ICASSP.2005.1416444}}.

\bibitem{Bell93}
M.~Bell, Information theory and radar waveform design, IEEE Transactions on
  Information Theory 39~(5) (1993) 1578 --1597.
\newblock \href {http://dx.doi.org/10.1109/18.259642}
  {\path{doi:10.1109/18.259642}}.

\bibitem{SYPMCR09}
S.~Sira, Y.~Li, A.~Papandreou-Suppappola, D.~Morrell, D.~Cochran,
  M.~Rangaswamy, Waveform-agile sensing for tracking, Signal Processing
  Magazine, IEEE 26~(1) (2009) 53 --64.
\newblock \href {http://dx.doi.org/10.1109/MSP.2008.930418}
  {\path{doi:10.1109/MSP.2008.930418}}.

\bibitem{BS2000}
H.~Boche, S.~Stanczak, Estimation of deviations between the aperiodic and
  periodic correlation functions of polyphase sequences in vicinity of the zero
  shift, in: IEEE Sixth International Symposium on Spread Spectrum Techniques
  and Applications, Vol.~1, 2000, pp. 283 --287.
\newblock \href {http://dx.doi.org/10.1109/ISSSTA.2000.878129}
  {\path{doi:10.1109/ISSSTA.2000.878129}}.

\bibitem{Nar08}
R.~Narayanan, Through wall radar imaging using uwb noise waveforms, in: IEEE
  International Conference on Acoustics, Speech and Signal Processing, ICASSP,
  2008, pp. 5185 --5188.
\newblock \href {http://dx.doi.org/10.1109/ICASSP.2008.4518827}
  {\path{doi:10.1109/ICASSP.2008.4518827}}.

\bibitem{NXHC98}
R.~M. Narayanan, Y.~Xu, P.~D. Hoffmeyer, J.~O. Curtis,
  \href{http://link.aip.org/link/?JOE/37/1855/1}{Design, performance, and
  applications of a coherent ultra-wideband random noise radar}, Optical
  Engineering 37~(6) (1998) 1855--1869.
\newblock \href {http://dx.doi.org/10.1117/1.601699}
  {\path{doi:10.1117/1.601699}}.
\newline\urlprefix\url{http://link.aip.org/link/?JOE/37/1855/1}

\bibitem{BellNar01}
D.~Bell, R.~Narayanan, Theoretical aspects of radar imaging using stochastic
  waveforms, IEEE Transactions on Signal Processing 49~(2) (2001) 394 --400.
\newblock \href {http://dx.doi.org/10.1109/78.902122}
  {\path{doi:10.1109/78.902122}}.

\bibitem{Karr93}
A.~F. Karr, Probability, Springer Texts in Statistics, Springer-Verlag, New
  York, 1993.

\bibitem{Chr03}
O.~Christensen, An {I}ntroduction to {F}rames and {R}iesz {B}ases,
  Birkh\"{a}user, 2003.

\bibitem{Dau92}
I.~Daubechies, Ten {L}ectures on {W}avelets, SIAM, 1992.

\bibitem{Hoeffding}
W.~Hoeffding, Probability inequalities for sums of bounded random variables,
  Journal of American Statistical Association 58~(301) (1963) 13 -- 30.

\bibitem{Bai99}
Z.~D. Bai, Methodologies in spectral analysis of large-dimensional random
  matrices, a review, Statist. Sinica 9~(3) (1999) 611--677.

\bibitem{Ledoux2001}
M.~Ledoux, The {C}oncentration of {M}easure {P}henomenon, Vol.~89 of
  Mathematical Surveys and Monographs, American Mathematical Society,
  Providence, RI, 2001.

\bibitem{LPRTJ2004}
A.~E. Litvak, A.~Pajor, M.~Rudelson, N.~Tomczak-Jaegermann, Smallest singular
  value of random matrices and geometry of random polytopes, Adv. Math. 195~(2)
  (2005) 491--523.

\bibitem{CanTao06}
E.~Cand{\`e}s, T.~Tao, Near optimal signal recovery from random projections:
  Universal encoding strategies?, IEEE Transactions on Information Theory
  52~(12) (2006) 5406 -- 5425.

\bibitem{Candes_06}
E.~Cand{\`e}s, Compressive sampling, in: Proceedings of the International
  Congress of Mathematicians, Vol. III, European Mathematical Society, Madrid,
  2006, pp. 1433--1452.

\bibitem{CRT06}
E.~Cand{\`e}s, J.~Romberg, T.~Tao, Robust uncertainty principles: Exact signal
  reconstruction from highly incomplete frequency information, IEEE
  Transactions on Information Theory 52~(2) (2006) 489 -- 509.

\bibitem{Don06}
D.~L. Donoho, Compressed sensing, IEEE Transactions on Information Theory
  52~(4) (2006) 1289 -- 1306.

\end{thebibliography}






\end{document}